\documentclass[10pt,sigconf,letterpaper]{acmart}

\newcommand*{\AcmFormat}{}%
\newcommand*{\CameraVersion}{}%

\newcommand*{\LongVersion}{}%

\ifdefined\LongVersion
\else
\newcommand*{\CutTable}{}%
\newcommand*{\CutText}{}%
\newcommand*{\CutSpace}{}%
\newcommand*{\NoSecTwo}{}%
\newcommand*{\CutProofs}{}%
\newcommand*{\CutBeyond}{}%
\newcommand*{\SingleExample}{}%
\fi

\usepackage{times}  
\ifdefined\AcmFormat
\else
\usepackage[hidelinks]{hyperref}

\usepackage{graphicx}
\usepackage{cite}
\usepackage{amssymb}
\usepackage{comment}
\usepackage{float}
\usepackage{subcaption}
\usepackage{url}
\usepackage{amsfonts}
\newtheorem{theorem}{Theorem}
\newtheorem{lemma}[theorem]{Lemma}
\newtheorem{definition}{Definition}

\fi
\usepackage{balance}

\ifdefined\AcmFormat
\else
\fi

\newcommand\blfootnote[1]{%
  \begingroup
  \renewcommand\thefootnote{}\footnote{#1}%
  \addtocounter{footnote}{-1}%
  \endgroup
}

\newcommand{\inband}{in-situ}
\newcommand{\Inband}{In-situ}
\newcommand{\DMT}{\Delta M(\tau)}
\newcommand{\CDMT}{C_M}
\newcommand{\DPM}{\Delta P_M(\tau)}

\newcommand{\CTO}{C_{\tau \Theta}}
\newcommand{\CPl}{\eta}

\newcommand{\Ov}{\Theta}

\usepackage{array}
\usepackage{subcaption}
\usepackage{enumitem}

\newlength{\grafflecm}
\makeatletter
\newenvironment{chapquote}[2][2em]
  {\setlength{\@tempdima}{#1}%
   \def\chapquote@author{#2}%
   \parshape 1 \@tempdima \dimexpr\textwidth-2\@tempdima\relax%
   \itshape}
  {\par\normalfont\hfill--\ \chapquote@author\hspace*{\@tempdima}\par\bigskip}
\makeatother

\ifdefined\AcmFormat

\ifdefined\CameraVersion

\ifdefined\LongVersion
\renewcommand\footnotetextcopyrightpermission[1]{} 
\else
\acmYear{2023}\copyrightyear{2023}
\acmConference[ANRW '23]{Applied Networking Research Workshop}{July 22--28, 2023}{San Francisco, CA, USA}
\acmBooktitle{Applied Networking Research Workshop (ANRW '23), July 22--28, 2023, San Francisco, CA, USA}
\acmPrice{15.00}
\acmDOI{10.1145/3606464.3606469}
\acmISBN{979-8-4007-0274-7/23/07}

\ccsdesc[500]{Networks~Network measurement}
\ccsdesc[500]{Networks~Network monitoring}
\ccsdesc[500]{Networks~OAM protocols}
\fi
\else
\renewcommand\footnotetextcopyrightpermission[1]{} 
\fi

\fi

\begin{document}

\makeatletter
\ifdefined\AcmFormat
\renewcommand\@formatdoi[1]{\ignorespaces}
\makeatother
\fi
\pagestyle{plain}

\date{}

\title{
The Observer Effect in Computer Networks
}

\ifdefined\AcmFormat
\ifdefined\BlindReview
\else

\author{Tal Mizrahi}
\affiliation{%
  \institution{Technion --- Israel Institute of Technology}
}

\author{Michael Schapira}
\affiliation{%
  \institution{Hebrew University of Jerusalem}
}

\author{Yoram Moses}
\affiliation{%
  \institution{Technion --- Israel Institute of Technology}
}
\fi
\else
\ifdefined\BlindReview
\else
\author{
{Tal Mizrahi\textsuperscript{\ensuremath\diamond}, Michael Schapira\textsuperscript{\ensuremath\dagger}, Yoram Moses\textsuperscript{\ensuremath\diamond}} \\ 
\textsuperscript{\ensuremath\diamond}Technion --- Israel Institute of Technology, \textsuperscript{\ensuremath\dagger}Hebrew University of Jerusalem
}

\fi
\fi

\thispagestyle{empty}

\ifdefined\AcmFormat
\begin{abstract}

Network measurement involves an inherent tradeoff between accuracy and overhead; higher accuracy typically comes at the expense of greater measurement overhead (measurement frequency, number of probe packets, etc.). Capturing the ``right'' balance between these two desiderata -- high accuracy and low overhead -- is a key challenge. However, the manner in which accuracy and overhead are traded off is specific to the measurement method, rendering apples-to-apples comparisons difficult. To address this, we put forth a novel analytical framework for quantifying the accuracy-overhead tradeoff for network measurements. Our framework, inspired by the \emph{observer effect} in modern physics, introduces the notion of a \emph{network observer factor}, which formally captures the relation between measurement accuracy and overhead. Using our ``network observer framework'', measurement methods for \textit{the same} task can be characterized in terms of their network observer factors, allowing for apples to apples comparisons. We illustrate the usefulness of our approach by showing how it can be applied to various application domains and validate its conclusions through experimental evaluation.
\ifdefined\LongVersion
\blfootnote{This paper is an extended version of~\cite{ObserverANRW}, published in the ACM Applied Networking Research Workshop (ANRW), 2024.}
\fi
\end{abstract}

\maketitle

\else
\maketitle

\section*{Abstract}

Network measurement involves an inherent tradeoff between accuracy and overhead; higher accuracy typically comes at the expense of greater measurement overhead (measurement frequency, number of probe packets, etc.). Capturing the ``right'' balance between these two desiderata--high accuracy and low overhead--is a key challenge. However, the manner in which accuracy and overhead are traded off is specific to the measurement method, rendering apples-to-apples comparisons difficult. To address this, we put forth a novel analytical framework for quantifying the accuracy-overhead tradeoff for network measurements. Our framework, inspired by the \emph{observer effect} in modern physics, introduces the notion of a \emph{network observer factor}, which formally captures the relation between measurement accuracy and overhead. Using our ``network observer framework'', measurement methods for \textit{the same} task can be characterized in terms of their network observer factors, allowing for apples to apples comparisons. To exemplify the usefulness of our approach, we show how it can be applied to various application domains and validate its conclusions through experimental evaluation.
\fi

\vspace{5mm}

\begin{chapquote}{\textit{Socrates}}
\hspace{-5mm}
The unexamined life is not worth living
\end{chapquote}

\vspace{-2mm}

\ifdefined\AddSpace\vspace{3mm}\fi
\ifdefined\AddSpaceInfocom\vspace{2mm}\fi
\section{Introduction}
\label{IntroSec}
\ifdefined\AddSpace\vspace{1mm}\fi
\ifdefined\AddSpaceInfocom\vspace{1mm}\fi
Network measurement is critical for detecting failures and anomalies, and also for verifying that desired performance bars, e.g., Service Level Agreements (SLAs), are met. Measurements can be realized in a variety of ways, including injecting periodic or on-demand control (probe) packets, continuously streaming telemetry information to central analyzers, and in-band collection of measurement data. 


\begin{sloppypar}
Network measurement involves an inherent tradeoff between accuracy of the measured metric (e.g., network delays, throughput, etc.) and the impact of measuring on network performance (e.g., the excess bandwidth consumed by probe packets, the increase in packet loss rate). Consider, for example, the challenge of guaranteeing that network delay within the premises of a service provider is below a certain value. To accomplish this, service providers typically rely on periodic delay measurements. For instance, a service provider that operates a network running MPLS-TP (Multi-Protocol Label Switch Transport Profile) can run periodic measurements using \emph{Delay Measurement Messages (DMM)}~\cite{rfc6374}. Naturally, to obtain accurate and timely measurements, DMM messages should be sent at high frequency. However, too high a frequency (e.g., $300$ messages per second) can entail prohibitively high communication overhead (and even affect network delay!). In this example, the measurement accuracy can be quantified in terms of how long it takes to detect that network delays exceed the threshold, whereas the performance overhead is the increase in data rate on the wire induced by sending DMM messages.

In this paper, we show that for a broad variety of measured metrics and notions of measurement overhead (including those of the above example), the accuracy-overhead tradeoff for any measurement method can be captured by a parameter that we call the \emph{network observer factor} induced by this method. Informally, the network observer factor $\CPl$ (inspired by the observer factor from modern physics) is a \textit{method-specific} constant that quantifies the measurement overhead per time units. The network observer factor provides a powerful conceptual framework for \textit{characterizing} the efficiency of a measurement method, and can serve as a useful metric for \textit{comparing} measurement methods for the same task. Going back to the DMM example, suppose that the service provider is considering two possible delay measurement protocols for MPLS-TP with similar functionality, as in~\cite{rfc6374,G8113.1}. As shall be explained below, by deriving the observer factors for the two protocols, the service provider can perform an `apples-to-apples' comparison of the two protocols.  
\end{sloppypar}

\ifdefined\AddSpace\vspace{3mm}\fi

\begin{sloppypar}
Our contributions are summarized below:

\begin{itemize}

\item We propose a formal framework for reasoning about the  accuracy-overhead tradeoff of a measurement method (the ``network observer effect'') and for characterizing measurement methods in terms of their ``network observer factor''. We view this theoretical framework as a powerful lens for evaluating the efficiency of measurement methods and believe that complementing our approach with the prevalent experimental/empirical analysis can facilitate deeper insights into \emph{inherent} tradeoffs between accuracy and overhead. 

\item We formally analyze the theoretical properties of the network observer effect and its associated factor, shedding light on the relation between performance measurement accuracy and measurement rate. In particular, we identify broad classes of measurement metrics and notions of measurement overhead for which the network observer factor can be analytically derived.

\item We validate the usefulness of our approach, as well as our analytical observations, by experimentally analyzing several measurement protocols of interest, spanning passive, active, and in-band measurements. In particular, we show how by fixing the level of measurement accuracy, our methodology facilitates apples-to-apples comparison of the efficiency of different measurement methods. 
    
\end{itemize}
\end{sloppypar}

\ifdefined\CutProofs
\ifdefined\BlindReview
Due to space limits, proofs and detailed discussion are deferred to the extended version of this paper~\cite{ObserverTechnicalBlind}.
\else
Due to space limits, proofs and detailed discussion are deferred to the extended version of this paper~\cite{ObserverTechnical}.
\fi
\fi

\ifdefined\AddSpace\vspace{2mm}\fi
\ifdefined\AddSpaceInfocom\vspace{2mm}\fi
\section{Inspiration and Overview}
\ifdefined\AddSpace\vspace{1mm}\fi
\ifdefined\AddSpaceInfocom\vspace{1mm}\fi
\subsection{The Observer Effect in Physics}
Heisenberg's uncertainty principle states that the position and momentum of a particle cannot both be measured precisely. 
This principle is often expressed by the following uncertainty relation~\cite{heisenberg1930physical}:

\begin{equation}
\label{eq:HeisenbergUncertainty}
\Delta x \cdot \Delta p \geq \hbar
\end{equation}

where $\Delta x$ is the level of uncertainty regarding a particle's position, $\Delta p$ is the uncertainty regarding its momentum, and~$\hbar$ is the Planck constant. 

Heisenberg argued~\cite{heisenberg1930physical} that a Gamma ray that is used to measure the particle's location will affect the particle's momentum. In other words, the measurement procedure affects the measured system--a phenomenon that later became known as the \emph{observer effect}.

Heisenberg's uncertainty principle has been found to be valid even if the system is not measured by an observer, and thus the uncertainty principle and the observer effect are regarded as two distinct principles of quantum mechanics. Notably, the uncertainty relation (Eq.~\ref{eq:HeisenbergUncertainty}) is applicable to the analysis of both principles.



\ifdefined\AddSpace\vspace{2mm}\fi
\ifdefined\AddSpaceInfocom\vspace{2mm}\fi
\subsection{The Network Observer Effect}
\ifdefined\AddSpace\vspace{1mm}\fi
\ifdefined\AddSpaceInfocom\vspace{1mm}\fi

Communication networks can also be regarded as exhibiting an observer effect, in the sense that the act of measuring the performance of a network often affects network performance. Inspired by Heisenberg's analysis of the uncertainty relation (Eq.~\ref{eq:HeisenbergUncertainty}), we introduce an analogous uncertainty relation.

\vspace{3mm}
\textbf{The uncertainty relation in networks:}
\begin{equation}
\label{eq:Uncertainty1}
\Delta M \cdot \Delta P \geq \CPl 
\end{equation}




In the above equation, $\Delta M$ is the uncertainty in the measured metric $M$ and $\Delta P$ is the change in a network performance metric, $P$, which is affected by the measurement. Intuitively, $\Delta M$ is the measurement's \emph{uncertainty} (or accuracy), whereas $\Delta P$ is the measurement's \emph{impact}, which quantifies the effect of the measurement on network performance. $\Delta P$ can have different meanings in different contexts, e.g., excess packet loss or delay induced by measurements, or excess data transfer charges induced by measurement probes. Naturally, lower uncertainty comes at the cost of higher measurement overhead. Going back to our \emph{DMM} example, suppose that our goal is to minimize the time it takes for the node that runs the periodic DMM to detect unusual congestion or failure. 
To reduce measurement uncertainty in terms of the time it takes to detect anomalous network behavior ($\Delta M$), the frequency of the DMM messages should be increased. This, however, would impact the performance by increasing the data rate on the wire (with an excess data rate of $\Delta P$), as the uncertainty relation suggests. 

Like Heisenberg's uncertainty relation, the network uncertainty relation clearly does not apply to all possible performance metrics and all possible measurement methods. In Section~\ref{AnalyzingSec} we characterize the conditions under which the uncertainty relation is applicable. 

Eq.~\ref{eq:Uncertainty1} captures the natural tradeoff between measurement accuracy and overhead; decreasing the uncertainty in performance measurement causes the impact on the network performance to increase. This tradeoff is illustrated in Fig.~\ref{fig:UncertaintyRelation}, which presents the impact as a function of the measurement uncertainty in a specific scenario. The curve represents the theoretical lower bound of the measurement impact according to Eq.~\ref{eq:Uncertainty1}. Our experimental evaluation (Section~\ref{EvalSec}) includes experiments that show that the measured impact nearly coincides with the theoretical lower bound of the uncertainty relation.

\begin{figure}[htbp]
	\centering
  \fbox{\includegraphics[trim={6cm 2cm 0 0},clip,height=7\grafflecm]{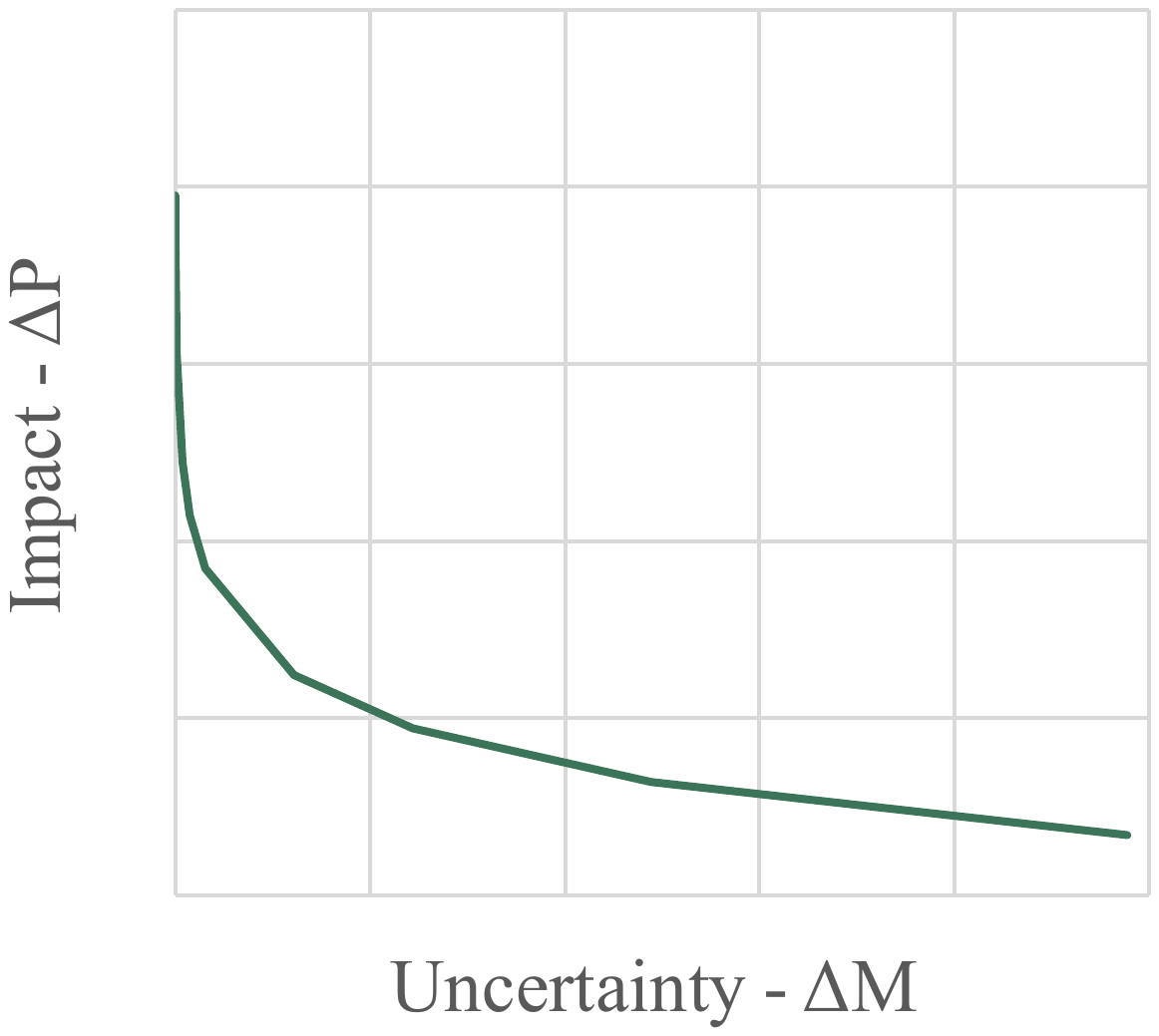}}
	\captionsetup{justification=centering}
  \caption{The impact as a function of the measurement uncertainty. The curve represents the lower bound, where  ${\Delta M \cdot \Delta P = \CPl}$.}
  \label{fig:UncertaintyRelation}
\ifdefined\AddSpace\vspace{2mm}\fi
\ifdefined\CutSpaceInfocom\vspace{-3mm}\fi
\ifdefined\AddSpaceInfocom\vspace{2mm}\fi
\end{figure}

$\CPl$ is a constant we call the \emph{network observer factor}, which quantifies the measurement overhead per time unit. Importantly, unlike the global Planck constant, $\CPl$ is \emph{specific} to the measurement method. $\CPl$ can be used as a metric for the efficiency a measurement method. Indeed, we contrast different measurement methods in terms of their associated $\CPl$ constants to gain insight into the relation between measurement accuracy and overhead. In our experimental evaluation (see Section~\ref{EvalSec}), we validate our theoretical findings and demonstrate the usefulness of our approach. We accomplish this by evaluating the considered measurement methods under similar uncertainty in measurement accuracy ($\Delta M)$ and observing the impact on network performance induced by each method ($\Delta P$), allowing for an apples-to-apples comparison. Interestingly, our theoretical analysis and evaluation show that, 
in the context of network delay measurements, in-situ methods such as In-Network Telemetry (INT), which seemingly require high overhead, have low uncertainty, and thus have almost identical impact on performance as other measurement methods (e.g., passive measurements) under the same level of achieved measurement accuracy.

\ifdefined\NoSecTwo
\else
\begin{sloppypar}
\ifdefined\AddSpace\vspace{9mm}\fi
\section{Why the Observer Effect Matters}
\ifdefined\AddSpace\vspace{1mm}\fi
The network observer effect, as presented above, applies to various types of networks: data centers, wide area and carrier networks, campus networks, and even home networks. We argued that a measurement method that has low uncertainty has a high performance impact on the network. One could argue that the fast growth of large-scale networks has opened the door to overprovisioned network resources (e.g.,~\cite{microsoft}), where the overhead of network measurement may arguably be negligible, or the measurement accuracy could be relaxed.

We present a few crisp use cases that demonstrate the observer effect in high-speed networks. The use cases demonstrate why highly accurate (detailed) measurement is important, and why overprovisioning does not necessarily avoid the measurement impact.
\end{sloppypar}

\subsection{Use Case 1: In-situ Measurement}
\ifdefined\AddSpace\vspace{1mm}\fi
In-situ measurement (also known as in-band or in-network telemetry) provides fine-grained and detailed monitoring and measurement by having every switch or router push telemetry information into the header of every data packet. Both research-driven and industry-driven protocols have been defined in this context (e.g.,~\cite{IOAM,INT,kumar-ippm-ifa-01,kimband}). In-situ measurement allows very detailed information about the path taken by every packet, the performance of switches along the path, and potential failures. However, the obvious penalty is the per-packet overhead which may cost tens of bytes per packet. In-situ measurement is an important example of a case where despite the high measurement overhead, this approach is increasingly gaining attention from the networking community due to the fine-grained information it provides, including at least one publicly known deployment~\cite{li2019hpcc}.

\subsection{Use Case 2: Broadband Home Access}
\ifdefined\AddSpace\vspace{1mm}\fi
The cost structure of home broadband subscriptions is bandwidth-sensitive, and therefore allocating a fraction of the bandwidth to measurement and monitoring would be frowned upon by home subscribers. The overprovisioning approach that has become common in public cloud networks~\cite{microsoft} is obviously not feasible in home networks. Therefore, an accurate measurement is likely to have noticeable impact on the performance of a home subscriber network.
	
\subsection{Use Case 3: Lossless Service}
\ifdefined\AddSpace\vspace{1mm}\fi
Consider a data center network in which storage-related traffic is forwarded as a lossless service. Lossless delivery is guaranteed by assigning a dedicated traffic class to this service, and by using Priority Flow Control (PFC).

\begin{sloppypar}
The network operator measures the lossless service, as even infrequent packet losses should be monitored and troubleshooted. The operator may occasionally encounter a few packet drops, caused by short micro-bursts that temporarily fill the queues. The measurement should be accurate and fine-grained, allowing to detect even a small number of drops. At the same time, the measurement overhead becomes significant when the queues are full or nearly full, since it then causes more packet drops than would be caused if the flow had not been measured. As shown in Fig.~\ref{fig:ObservedVsUnobserved}, when one of the links reaches its capacity, even a small increase in the traffic rate will affect the number of packet drops. This use case demonstrates why even the slightest measurement impact may in some cases be noticeable in the network we aim to measure.

\begin{figure}[htbp]
	\centering
  \fbox{\includegraphics[height=7\grafflecm]{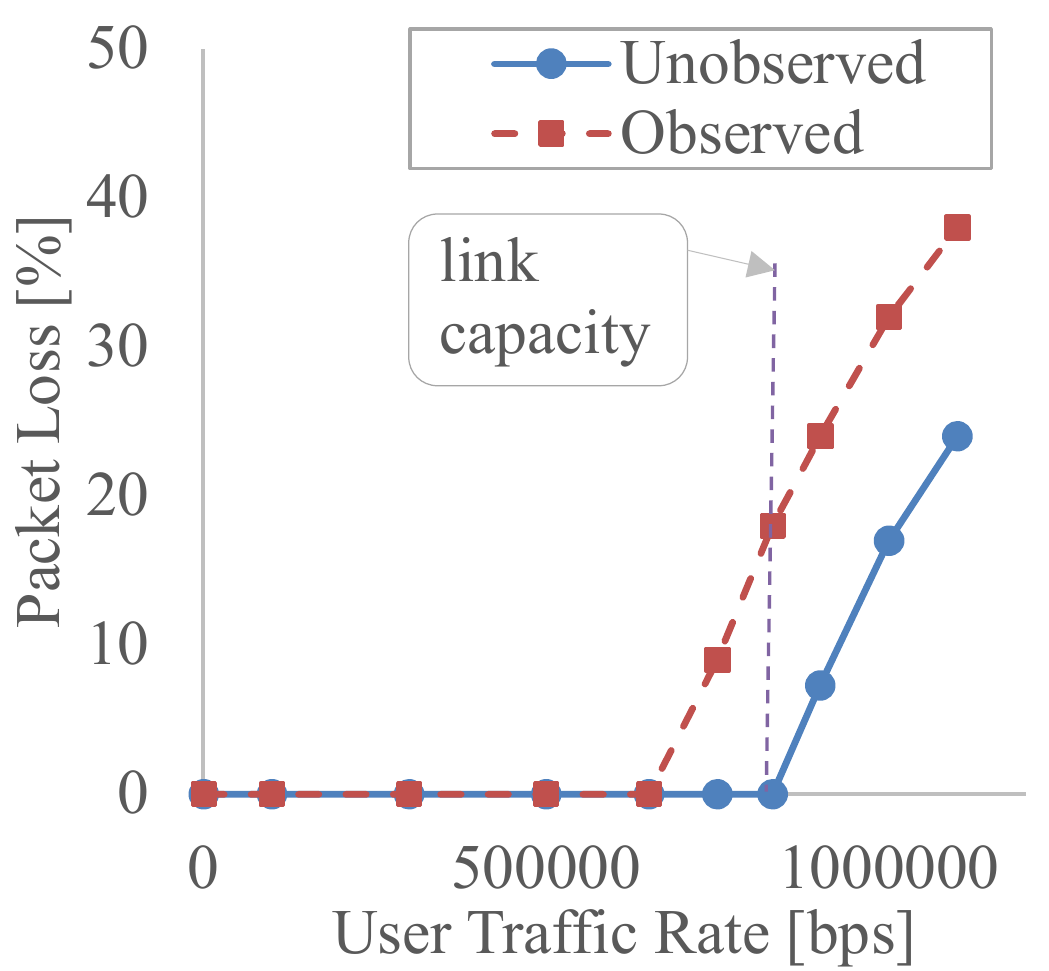}}
	\captionsetup{justification=centering}
  \caption{Experimental example of the network observer effect: an observed flow (monitored by IOAM~\cite{IOAM}) has a higher packet loss rate than unobserved flow. Fixed user traffic rate.}
  \label{fig:ObservedVsUnobserved}
\ifdefined\AddSpace\vspace{2mm}\fi
\ifdefined\CutSpaceInfocom\vspace{-3mm}\fi
\ifdefined\AddSpaceInfocom\vspace{2mm}\fi
\end{figure}

\subsection{Use Case 4: Service-Level Agreement (SLA) Measurement}
\ifdefined\AddSpace\vspace{1mm}\fi
A Service-Level Agreement (SLA) in carrier networks is an agreement between a customer and a service provider regarding the performance of a network service. As defined in the MEF 10.3 specification~\cite{MEF10.3}, a Service Level Specification (SLS) defines the performance objectives for a given bandwidth. I.e., the SLS specifies the agreed delay, delay variation, and loss rate that are guaranteed as long as the customer does not exceed an agreed bandwidth.
\end{sloppypar}

If a customer is interested in measuring the network, in order to guarantee that the SLS is satisfied, the measurement overhead effectively decreases the user's bandwidth compared to the bandwidth that is specified in the SLS.

\fi

\begin{sloppypar}
\ifdefined\AddSpace\vspace{2mm}\fi
\section{Analyzing the Observer Effect}
\label{AnalyzingSec}
\ifdefined\AddSpace\vspace{1mm}\fi

We now investigate the network uncertainty relation, which is inspired by Heisenberg's uncertainty relation and its connection to the observer effect. We first present the terminology, assumptions, and the model used in the analysis.

\ifdefined\AddSpace\vspace{2mm}\fi
\subsection{Measurement Classes}
\label{ClassSec}
\ifdefined\AddSpace\vspace{1mm}\fi
We analyze the measurement granularity in terms of three main measurement classes~\cite{rfc7799}: passive, active, and hybrid measurement.

\subsubsection{Passive Measurement}
Passive methods observe network traffic without modifying it and without transmitting control packets along the data path. Notably, since measurements are performed locally by each network node, measurement data is typically exported to a central aggregator. Measurement data may be exported periodically, on-demand, or triggered by specific events. 
A common protocol used for exporting passive measurement results is gNMI~\cite{gnmi}, which can be used by network devices to stream telemetry information to central collectors.

\subsubsection{Active Measurement} 
Active methods use synthetic traffic to measure the network. \emph{Ping} and \emph{Traceroute} are common examples. Home network speed tests (e.g.,~\cite{Speedtest}) perform the measurement by temporarily running synthetic traffic at a high rate, while other methods use periodic control messages, such as Continuity Check Messages (CCM)~\cite{IEEE802.1ag,Y1731}, to continuously monitor the network. 

To analyze the overhead of active measurement we focus on periodic measurement, which continuously monitors the traffic. This allows to detect a problem once it occurs, but obviously comes at the cost of continuous overhead.

\subsubsection{\Inband\ measurement}
\Inband\ (or in-band) methods, such as 
In-band Network Telemetry (INT)~\cite{kimband} and In situ OAM (IOAM)~\cite{IOAM},
piggyback measurement data onto live user packets.\footnote{Note that Direct Exporting~\cite{dex} or Postcard Mode~\cite{TelemetryReport} are forms of passive measurement in our context, as telemetry data is not forwarded along the data path.} This data is peeled off at a decapsulation node, which also exports some (or all) of the measurement data to a central collector.

Note that \inband\ measurement is one example of the hybrid methods defined in~\cite{rfc7799}. Other hybrid methods are also defined there, but strictly from a measurement \emph{overhead} perspective, which is what the current paper is focused on, each of these other methods belongs to one of the three classes above.

\vspace{7mm}
Notably, we should distinguish between data path overhead and management overhead. Data path overhead is overhead that is induced along the path of the user traffic that is being measured. In contrast, management overhead is required for exporting information to a central collector; the export path may or may not coincide with the data path. 
Active and \inband\ measurements incur data path overhead, since for these methods the overhead occurs along the same path as the user traffic. Furthermore, these two methods typically also entail passive overhead, derived from allowing the measurement information to be exported to a collector. Passive measurement entails only management overhead. 


In each of the measurement classes, the overhead is quantified in a different way. Active measurement uses periodic control messages, and thus the overhead is on a per-time-unit basis. The overhead of \inband\ measurement is quantified on a per-packet basis. Passive measurement requires only management overhead, which we also consider on a per-packet basis.

\ifdefined\CutTable
\else
\ifdefined\CutSpace\vspace{-3mm}\fi
\ifdefined\AddSpace\vspace{2mm}\fi
\begin{table}[htbp]
		\centering
    \begin{tabular}{| p{1.8cm}<{\centering} || p{1.7cm}<{\centering} | p{1.7cm}<{\centering} | p{1.7cm}<{\centering} |}
    \hline
    & Passive & Active & \Inband\ \\ \hline \hline
    Data path overhead& & \checkmark & \checkmark \\ \hline 
    Management overhead& \checkmark & \checkmark & \checkmark \\ \hline 
    Overhead granularity& per time unit & per time unit & per packet \\ \hline 
    \end{tabular}
    \caption{Measurement classes.}
    \label{table:MeasurementClasses}
\end{table}
\ifdefined\AddSpace\vspace{1mm}\fi
\fi

\ifdefined\AddSpace\vspace{2mm}\fi
\subsection{Metrics}
\ifdefined\AddSpace\vspace{1mm}\fi
\end{sloppypar}

As discussed in Section~\ref{IntroSec}, there is a tight coupling between the desired measurement \emph{granularity} and the \emph{impact} of the observation on the measurement. A key question is which metrics, i.e., the $M$ and $P$ in Eq.\ref{eq:Uncertainty1}, are amenable to our type of analysis. 



The uncertainty relation is clearly not applicable to all possible performance metrics. In our analysis the measured metric $M$ is specifically a \emph{sensitive performance metric}, as defined in the following subsection. The impacted performance metric $P$ is a rate metric, i.e., a metric that is affected by the data rate or loss rate. Our analysis focuses on metrics that are impacted by the measurement regardless of the network utilization. For example, if $\Delta P$ represents the impact of the measurement on the traffic rate in the measured path or link, this impact is inherent in the measurement, regardless of how utilized the network is. Consequently, the analysis of the observer effect is not affected by overprovisioning. We prove below that the uncertainty relation is applicable in this context.

\ifdefined\AddSpace\vspace{2mm}\fi
\subsection{Model and Definitions}
\label{Model}
\ifdefined\AddSpace\vspace{1mm}\fi
Our analysis assumes that measurements are performed for a specific traffic flow or set of flows, where a flow consists of a set of packets with common characteristics, such as 5-tuple properties. The \emph{data rate} of a traffic flow in a network is defined to be the number of bits per time unit that are successfully delivered from the source to the destination, including the data plane headers and any overhead that may be incurred by the performance measurement. The \emph{loss rate} of a flow is defined to be the number of bits per time unit that are sent by the source and not delivered to the destination. It is also assumed that the data rate of the analyzed flow(s) is constant.\footnote{This assumption simplifies the analysis of the uncertainty relation. In practice, the data rate is not necessarily constant, but the analysis can be performed in sufficiently short time intervals, in which the data rate is roughly constant.}

Table~\ref{Notations} summarizes the notation used in this section.

\ifdefined\AddSpace\vspace{2mm}\fi
\begin{table}[!h]
		\centering
    \begin{tabular}{| l | p{6.5cm}|}
    \hline
	  $M$ & A sensitive performance metric that is measured periodically. \\
	  $\tau$ & The measurement period. \\
	  $\DMT$ & The amount of uncertainty in $M$ when measured periodically with a period $\tau$. \\
	  $P$ & A rate metric. \\
	  $\DPM$ & The measurement impact on $P$ when $M$ is measured periodically with a period $\tau$. \\
		$\Ov$ & The measurement overhead, measured in bits per time unit. \\
	  \hline
    \end{tabular}
    \caption{Notations}
    \label{Notations}
\end{table}
\ifdefined\AddSpace\vspace{1mm}\fi

A \emph{measurement} in our context is a process that observes a traffic flow using a fixed measurement method, which requires an overhead of $\Ov$ bits per time unit. The overhead $\Ov$ and the amount of uncertainty in the performance metrics $M$ and $P$ depend on the specific measurement method. As discussed in Section~\ref{ClassSec}, in different measurement classes the overhead may have a different impact on $P$, which may in turn affect the data path or the management path (or both). 

\begin{definition}[Sensitive performance metric]
A performance metric is said to be \emph{sensitive} if the amount of uncertainty when the metric is measured periodically with a period $\tau$ can be represented by a function $\DMT$ such that $\DMT = \CDMT \cdot \tau$ for some constant $\CDMT$.
\end{definition}

\begin{sloppypar}
\begin{definition}[Uncertainty in a measured metric]
Let $M$ be a performance metric that is measured periodically with a period $\tau$, and $M(t)$ be the measured value at time $t$. The uncertainty $\DMT$ is the minimal value for which $|M(t')-M(t)|\leq \DMT$ for all $t'$ such that $t<t'\leq t+\tau$.
\end{definition}

\begin{definition}[Rate metric]
A metric that either quantifies the traffic rate or the traffic loss rate is called a \emph{rate metric}.
\end{definition}

\begin{definition}[Measurement impact]
We define the measurement impact $\DPM$ of a rate metric $P$ when a sensitive metric $M$ is measured with a period $\tau$ to be the minimal upper bound on the difference between the value of the metric $P$ when $M$ is measured and the value of the metric $P$ without the measurement.
\end{definition}
\end{sloppypar}

\ifdefined\AddSpace\vspace{2mm}\fi
\subsection{The Network Uncertainty Relation}
\ifdefined\AddSpace\vspace{1mm}\fi

Our analysis starts with a basic claim (Lemma~\ref{RateMetricLemma}) about the connection between the measurement impact $\DPM$ and the measurement overhead $\Ov$.

\begin{lemma}
\label{RateMetricLemma}
If a metric $M$ is measured periodically with a period $\tau$ for a given flow, the mean measurement overhead per time unit is $\Ov$, and $P$ is a rate metric of the flow, then the impact $\DPM$ satisfies $\DPM \geq \Ov$.
\end{lemma}

\ifdefined\CutProofs
The proof of Lemma~\ref{RateMetricLemma} and of all technical claims that follow can be found in~\cite{ObserverTechnical}.
\else
\begin{proof}
A rate metric $P$ may be either a loss rate metric or a data rate metric. For the given measured metric $M$ and period $\tau$, we denote the loss rate uncertainty by $\Delta L_M(\tau)$ and the data rate uncertainty by $\Delta D_M(\tau)$. The loss rate uncertainty results from the measurement overhead, i.e., the highest loss rate occurs when the overhead exceeds the flow bandwidth, causing a loss rate of $\Ov$. Thus the difference between the maximal loss and the minimal loss satisfies ${\Delta L_M(\tau) \geq \Ov}$. The data rate uncertainty results from the fact that measurement (overhead) traffic may either be lost or not, depending on the network utilization (or overprovisioning), and therefore $\Delta D_M(\tau)\geq \Ov$. 
\end{proof}
\fi

The following theorem is the networking variant of Heisenberg's uncertainty relation. Intuitively, the theorem captures the tradeoff between the uncertainty in a measured performance metric $M$ and its effect on a corresponding rate metric $P$.

\begin{theorem}
\label{UncertaintyTheorem}
If $M$ is a sensitive performance metric of a flow that is measured periodically with a period $\tau$, and $P$ is a rate metric of the flow, then there exists a constant $\CPl$ such that:
\begin{equation}
\label{eq:UncertaintyMain}
\DMT \cdot \DPM \geq \CPl
\end{equation}
\end{theorem}

\ifdefined\CutProofs
\else
\begin{proof}
Since $M$ is a sensitive performance metric, then by the definition of a sensitive metric there exists a constant $\CDMT$ such that $\DMT = \CDMT \cdot \tau$. Therefore $\DMT \cdot \DPM = \CDMT \cdot \tau \cdot \DPM$. By Lemma~\ref{RateMetricLemma} $\DPM \geq \Ov$, and thus $\CDMT \cdot \tau \cdot \DPM \geq \CDMT \cdot \tau \cdot \Ov$. 
Since a fixed measurement method was assumed, with an overhead of $\Ov$ bits per time unit, it follows that $\tau \cdot \Ov$ is constant; we denote this constant by $\CTO$. Thus, $\CDMT \cdot \tau \cdot \Ov = \CDMT \cdot \CTO$. We define $\CPl$ to be $\CDMT \cdot \CTO$, and thus obtain $\DMT \cdot \DPM \geq \CPl$.
\end{proof}
\fi

\begin{figure*}[htbp]
  \centering
  \begin{subfigure}[t]{.33\textwidth}
  \centering
  \fbox{\includegraphics[height=6.5\grafflecm]{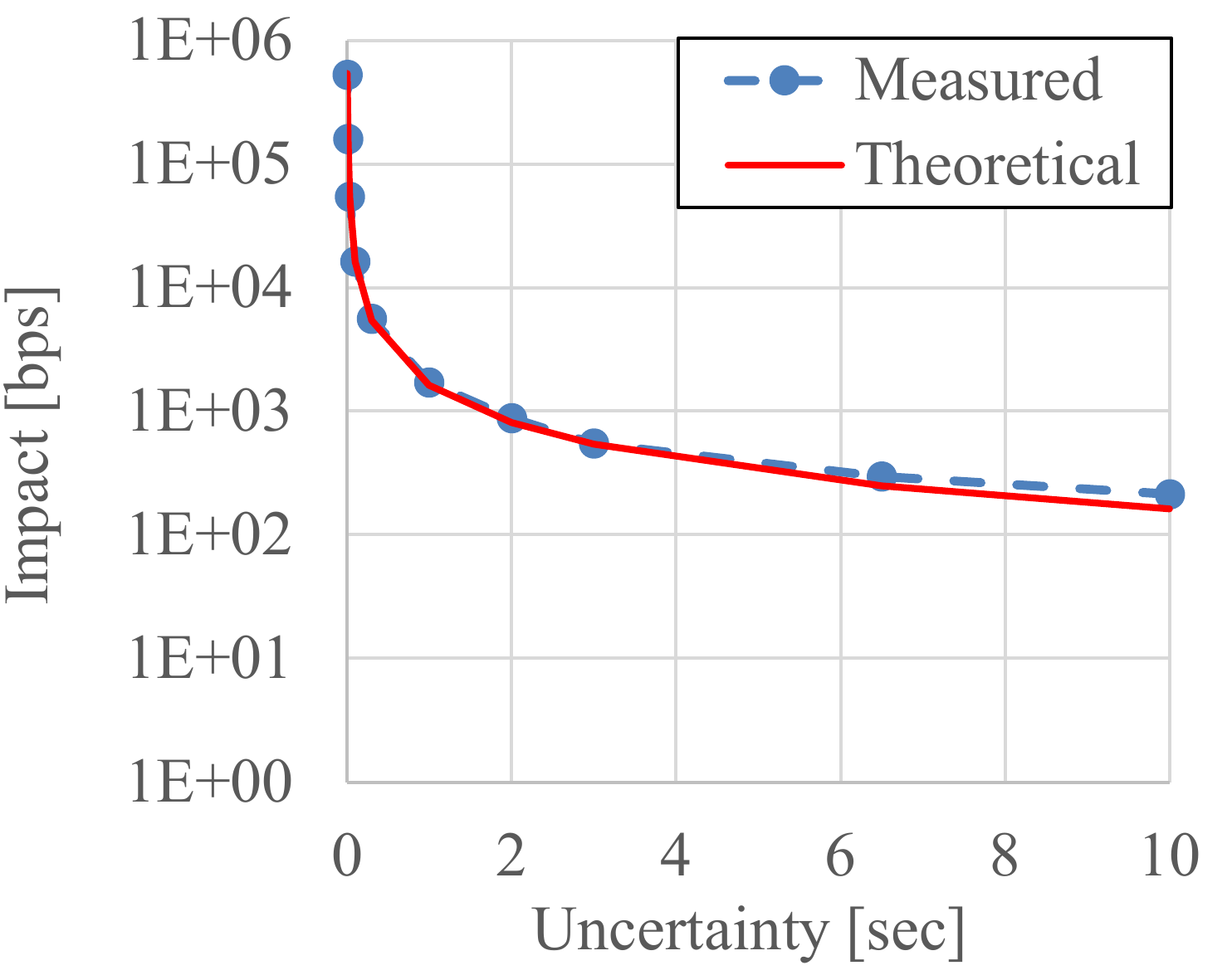}}
  \captionsetup{justification=centering}
  \caption{Passive measurement \\ and exporting using gNMI.}
  \label{fig:gnmi}
  \end{subfigure}%
  \begin{subfigure}[t]{.33\textwidth}
  \centering
  \fbox{\includegraphics[height=6.5\grafflecm]{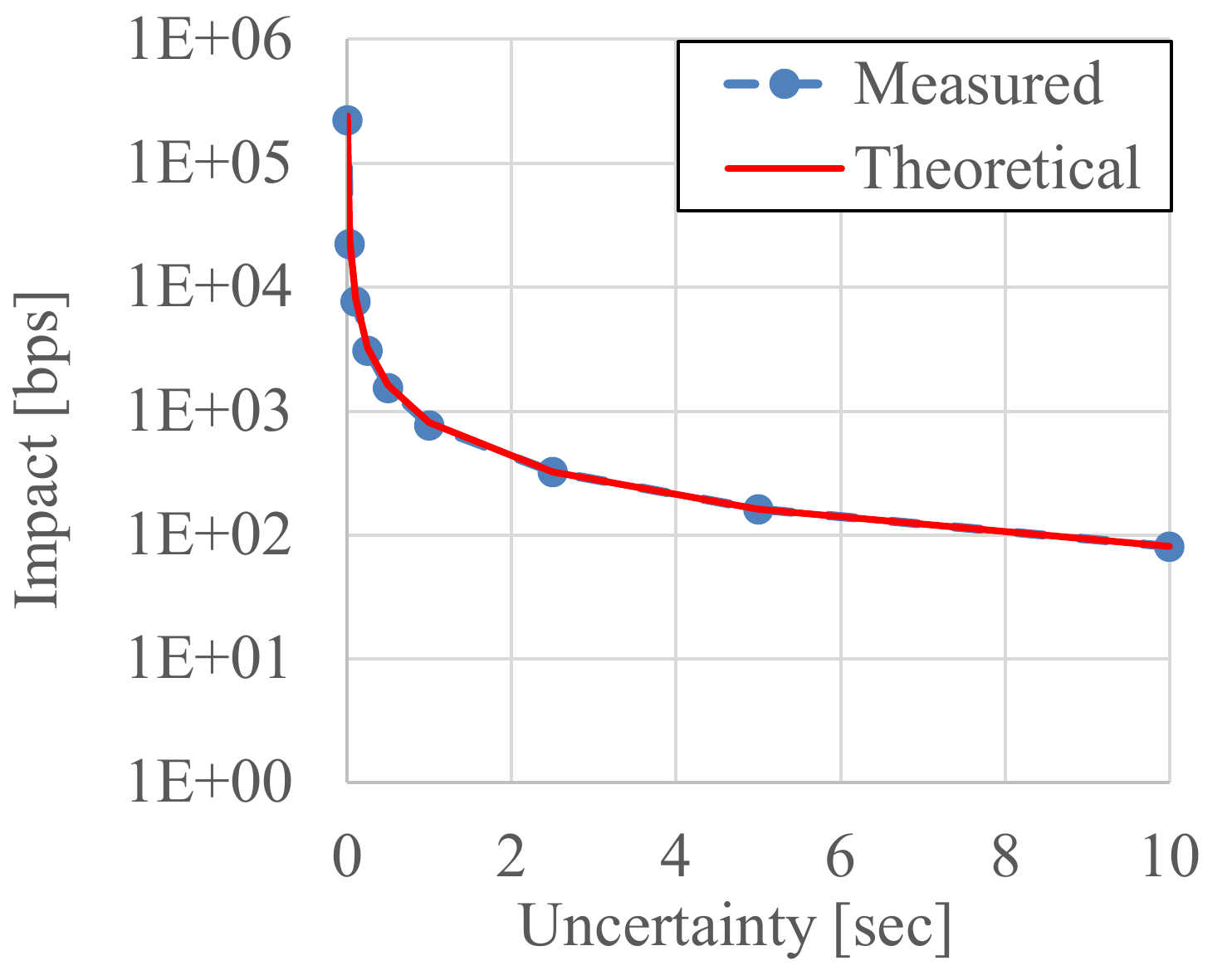}}
  \captionsetup{justification=centering}
  \caption{Active measurement using CCM.}
  \label{fig:ccm}
  \end{subfigure}%
  \begin{subfigure}[t]{.33\textwidth}
  \centering
  \fbox{\includegraphics[height=6.5\grafflecm]{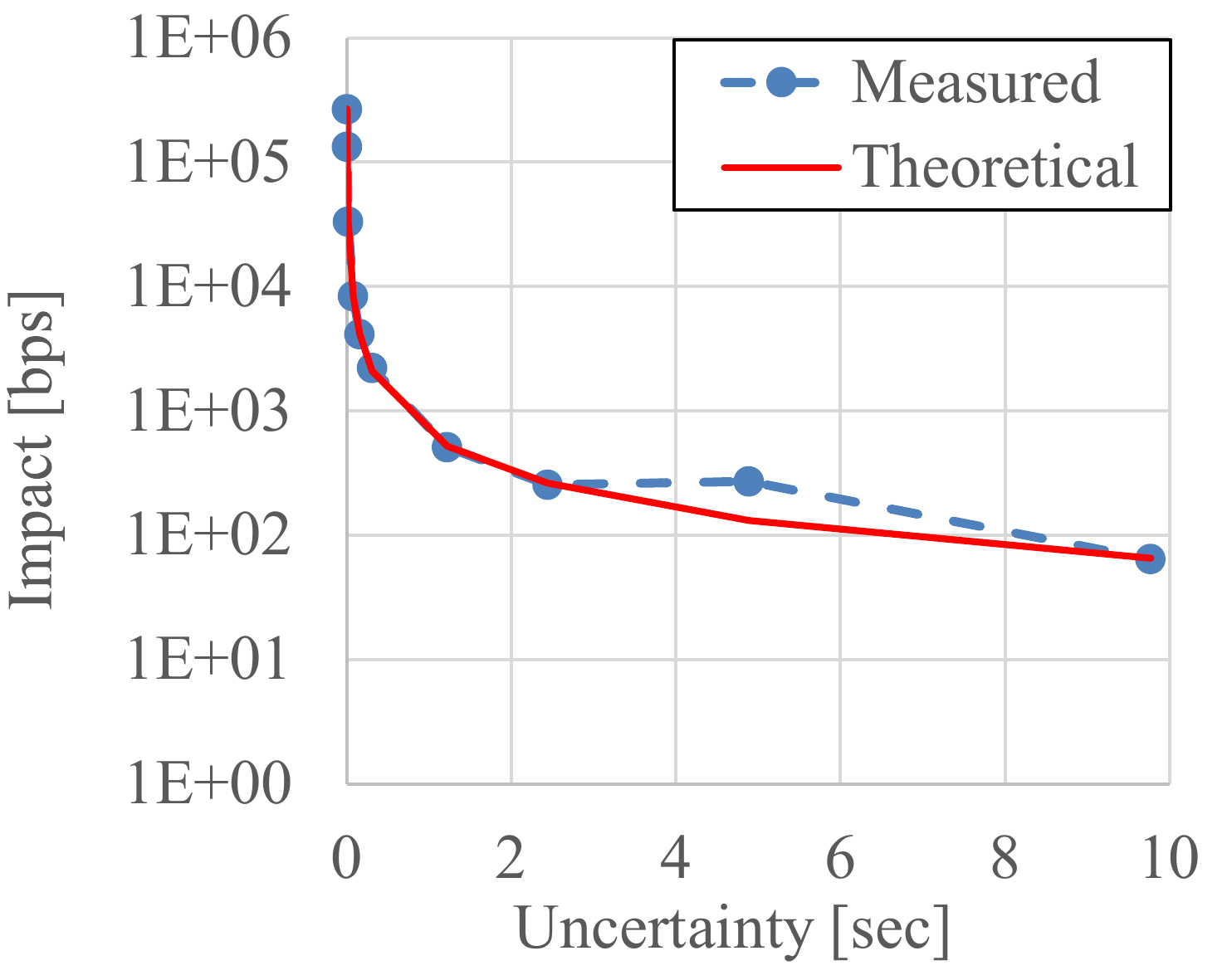}}
  \captionsetup{justification=centering}
  \caption{In situ measurement using IOAM.}
  \label{fig:ioam}
  \end{subfigure}%
  \caption{The uncertainty relation (Theorem~\ref{UncertaintyTheorem}) in practice: the measurement impact vs. the measurement uncertainty. For each of the three measurement classes the experimental result is compared to the theoretical result (predicted by the uncertainty relation).}
  \ifdefined\CutSpace\vspace{-3mm}\fi
  \label{fig:overhead}
\end{figure*}

\ifdefined\SingleExample
To understand the importance of Theorem~\ref{UncertaintyTheorem} we present the following example.

\textbf{Example 1.} Continuity Check Messages (CCM) are used in Ethernet OAM~\cite{IEEE802.1ag,Y1731} in order verify the continuity of an Ethernet link and to detect failures. CCMs are sent periodically, and a failure is reported when a CCM has not arrived within a given timeout. If we define the measured metric $M$ to be the detection time of a failure, then the frequency of the CCMs determines the uncertainty in the detection time $M$. The tradeoff between the uncertainty in the detection time and the impact of the measurement on the rate is captured in Theorem~\ref{UncertaintyTheorem}. In this case $\DMT=\tau$, and it is easy to see that $\CPl$ is equal to the number of overhead bits per period $\tau$.

\else

\ifdefined\AddSpace\vspace{2mm}\fi
\subsection{Understanding the Network Uncertainty Relation}
\ifdefined\AddSpace\vspace{1mm}\fi

In order to understand the impact of Theorem~\ref{UncertaintyTheorem} we present three examples, in the context of the three measurement classes that were discussed in Section~\ref{ClassSec}.

\textbf{Example 1.} Continuity Check Messages (CCM) are used in Ethernet OAM~\cite{IEEE802.1ag,Y1731} in order verify the continuity of an Ethernet link and to detect failures. CCMs are sent periodically, and a failure is reported when a CCM has not arrived within a given timeout. If we define the measured metric $M$ to be the detection time of a failure, then the frequency of the CCMs determines the uncertainty in the detection time $M$. Once again, the tradeoff between the uncertainty in the detection time and the impact of the measurement on the rate is captured in Theorem~\ref{UncertaintyTheorem}. In this case $\DMT=\tau$, and it is easy to see that $\CPl$ is equal to the number of overhead bits per period $\tau$.

\textbf{Example 2.} Consider a passive measurement process, in which a performance metric $M$ such as a flow counter is measured by a network device and periodically exported to a collector node. The measurement period $\tau$ is an indication of the freshness of the information that the collector obtains from the measured device; the data is fresh immediately after the measurement, but any change in the measured metric after the measurement is only known to the collector upon receiving the next measurement. Thus, $\DMT$ represents the maximal difference between two consecutive measurement values of $M$, and is directly proportional to $\tau$. Following Theorem~\ref{UncertaintyTheorem}, any reduction in the uncertainty of $M$ will result in an increase in the impact on the loss and/or data rate.

\textbf{Example 3.} Consider a flow in which IOAM is used for monitoring the network path of the flow. A collector is used to monitor the IOAM data, and detect when the network path changes. The metric $M$ in this example refers to the number of packets that were forwarded in the flow before a path change occurred. As IOAM may be applied to all packets of the flow or to a subset of the packets, we can define~$\tau$ to be the mean period between two packets that are monitored by IOAM. Thus, by Theorem~\ref{UncertaintyTheorem}, there is a tradeoff between the metric~$M$ and the impact on the flow rate.

\fi

\ifdefined\AddSpace\vspace{2mm}\fi
\subsection{A Concrete Observer Factor}
\ifdefined\AddSpace\vspace{1mm}\fi
The network uncertainty relation (Eq.~\ref{eq:UncertaintyMain}) uses the constant $\CPl$, representing the tradeoff between the measurement's uncertainty and impact. This constant is specific to the measurement method, in contrast to $\hbar$ in Eq.~\ref{eq:HeisenbergUncertainty}, which is a universal constant. In Example~1 the factor $\CPl$ has an intuitive meaning and can be easily computed; it is simply the number of overhead bits per measurement period. The following lemma reflects this point. As in Example~1, we analyze the uncertainty in the \emph{detection time}, $T$, of a failure or anomaly.

\begin{sloppypar}
\begin{lemma}
\label{ConstantLemma}
Let $T(\tau)$ be the detection time in a periodic measurement with a period $\tau$, and let $P$ be a rate metric. Then ${\Delta T(\tau) \cdot \Delta P_T (\tau) \geq \CPl}$, where $\CPl$ is the number of overhead bits per period $\tau$ used by the measurement. 
\end{lemma} 
\end{sloppypar}

\ifdefined\CutProofs
\else
\begin{proof}
Since the measurement is periodic with a period $\tau$, a failure or anomaly is detected at most $\tau$ time units after it occurs, and thus $\Delta T(\tau) = \tau$. 
By Lemma~\ref{RateMetricLemma} we have that $\Delta P_T (\tau) \geq \Ov$. Thus, $\Delta T(\tau) \cdot \Delta P_T (\tau) \geq \tau \cdot \Ov$. Note that $\tau \cdot \Ov$ is the number of overhead bits per period, and we denote it by $\CPl$, yielding ${\Delta T(\tau) \cdot \Delta P_T (\tau) \geq \CPl}$.
\end{proof}
\fi

\ifdefined\LongVersion
\subsection{Scaling the Observer Effect}
An important question that arises with respect to the observer effect is whether it is still relevant at high scales. For example, one may argue that the impact of a periodic active measurement method, such as the CCM protocol, may be significant in low-bandwidth networks, but becomes insignificant in large-scale networks. We will show that in a precise sense, the measurement impact scales with the size of the network.

At a first glance it may seem that the desired uncertainty (or detection time), $\Delta T(\tau)$, does not scale with the network size.
The best-known requirement in telecom networks is the ability to recover from a failure within 50~milliseconds. This requirement has not changed in many years, as it is derived from the human ability to detect downtimes in voice calls. However, in large-scale data center networks much faster detection times may be required. For example, if we consider a high-bandwidth transaction between servers over a 100~Gbps network interface, a 50~millisecond detection time yields 5~Gigabits of data, which may be lost before an error is detected. Thus, the measurement uncertainty requirements become increasing more stringent as networks scale.
Moreover, even if the desired uncertainty is a fixed requirement that is determined by the application (as in the 50~millisecond example), we expect the measurement overhead and impact to scale with the number of flows.

This intuition is formalized in the following lemma:

\begin{sloppypar}
\begin{lemma}
\label{FlowLemma}
Let $T(\tau)$ be the detection time in a periodic measurement with a period $\tau$ that is performed for $N$ flows, and let $P$ be a rate metric. Then ${\Delta T(\tau) \cdot \sum_{i=0}^{N-1} \Delta {P_T}_i (\tau) \geq N \cdot \CPl}$, where $\CPl$ is the number of overhead bits per period $\tau$. 
\end{lemma} 

\begin{proof}
We observe that \\ 
${\Delta T(\tau) \cdot \sum_{i=0}^{N-1} \Delta {P_T}_i (\tau) = \sum_{i=0}^{N-1} \Delta T(\tau) \cdot \Delta {P_T}_i (\tau)}$. 
By Theorem~\ref{ConstantLemma} we obtain that \\ 
${\sum_{i=0}^{N-1} \Delta T(\tau) \cdot \Delta {P_T}_i (\tau) \geq \sum_{i=0}^{N-1} \CPl = N \cdot \CPl}$. It follows that ${\Delta T(\tau) \cdot \sum_{i=0}^{N-1} \Delta {P_T}_i (\tau) \geq N \cdot \CPl}$.
\end{proof}
\end{sloppypar}

This lemma provides an important insight about the scaling of the observer effect: for a fixed detection time, the impact scales with the size of the network.
\fi

\section{Evaluation}
\label{EvalSec}
\ifdefined\AddSpace\vspace{1mm}\fi

We now present an experimental evaluation of the network observer effect for each of the three measurement classes of Section~\ref{AnalyzingSec}: passive, active and \inband. We evaluate three significantly different protocols, gNMI, CCM and IOAM, in three different open source environments. 

\ifdefined\BlindReview
All of our evaluation results can be replicated using the detailed instructions and code in~\cite{git-oe-blind}.
\else
All of our evaluation results can be replicated using the detailed instructions and code in~\cite{git-oe}.
\fi

To allow for an `apples-to-apples' comparison, despite the inherent differences between the three classes (Section~\ref{ClassSec}), we assume that the measurement overhead is carried over the data path.\footnote{Specifically, in passive measurement this refers to the case in which the measurement data is exported along the same path as the data itself. Although our analysis does not mandate this assumption, it is convenient for an apples-to-apples comparison.} Specifically, in the IOAM measurement we assume that data flows have a constant data rate of 1 Mbps, using a packet length of $360$ bytes.\footnote{This is the average packet length in a simple IMIX~\cite{IMIX}, including IPv6 and L2 headers.} This assumption is specifically relevant to IOAM, where measurement overhead is a function of the data rate.

\ifdefined\AddSpace\vspace{2mm}\fi
\subsection{Experiment Setup}
\label{ExperimentSetupSec}
\ifdefined\AddSpace\vspace{1mm}\fi

Three experimental environments were used in the evaluation:

\textbf{gNMI.} The passive measurement experiment setup was based on an open source environment running a Stratum/BMv2 switch~\cite{stratum}, emulated in Mininet. Two hosts were connected through the switch, and a counter was used to monitor the traffic through each of the switch's interfaces. gNMI~\cite{gnmi} was used for periodically exporting the value of one of the counters from the switch; we used the gNMI client to subscribe to periodic updates from the switch, and varied the export interval from a few milliseconds to a few seconds. The impact (management overhead) of this telemetry stream was analyzed. 

\textbf{CCM.} The active measurement setup was based on an open source implementation~\cite{dot1ag-utils}
of the IEEE 802.1ag standard~\cite{IEEE802.1ag} for Ethernet Connectivity Fault Management. The active measurement procedure was performed by periodic Continuity Check Messages (CCM) that were sent between two hosts in a Mininet environment. We evaluated the measurement impact for various CCM interval values.

\textbf{IOAM.} We tested \Inband\ measurement using an open source implementation of IOAM in IPv6~\cite{IOAM,ioam-ipv6}. Traffic was sent between two hosts and forwarded along three hops of switches: the switch pushed an IPv6 tunnel with the IOAM encapsulation, the second was used as an IOAM transit switch and pushed IOAM data into transit packets, and the third pushed its own IOAM data and then removed the IPv6 tunnel and IOAM encapsulation. The two hosts and three switches were emulated by five Linux Containers (LXC).

\ifdefined\AddSpace\vspace{1mm}\fi
\subsection{The Impact-Uncertainty Tradeoff}
\ifdefined\AddSpace\vspace{1mm}\fi

In each of the three experiments we focused on the \emph{data rate} as the metric $P$. Assuming an overprovisioned network, the impact $\DPM$ is the difference between the data rate including the measurement and the data rate without the measurement. Thus, the impact is equal to the measurement overhead (this statement is confirmed by the evaluation in the next subsection). We define $\DMT$ to be the uncertainty in the detection time of a failure or anomaly, and by Lemma~\ref{ConstantLemma} we have that $\DMT=\tau$, and the constant $\CPl$ is equal to the number of overhead bits per period $\tau$.

We ran the experiments with various values of $\tau$. In gNMI the period $\tau$ is determined by the exporting interval, and similarly the CCM interval determines $\tau$ in the active case. In IOAM we varied the sampling ratio, which is the fraction of data traffic that is monitored by IOAM. Since the data rate in the experiment is constant, the sampling ratio determines the measurement interval $\tau$.

The experimental results are presented in Fig.~\ref{fig:overhead}. In each measurement setup we measured the impact $\DPM$ as a function of the uncertainty $\DMT$. Each graph also depicts the theoretical 
curve $\DPM = \CPl / \DMT$, illustrating the uncertainty relation (Eq.~\ref{eq:UncertaintyMain}).

The theoretical value of $\CPl$ is the expected overhead per period of each of the protocols. In gNMI, a telemetry message of $\CPl (gNMI)=204$ bytes was exported by the switch in each period. Each CCM is $\CPl (CCM)=101$ bytes long. IOAM was used over a three-hop network, where each hop pushed an overhead of 8~bytes, with a total of $\CPl (IOAM)=80$ bytes including the IPv6 tunnel header (44~bytes), the IPv6 option header (4~bytes) and the IOAM header (8~bytes), which were added by the encapsulating switch.


The results confirm the uncertainty relation of Theorem~\ref{UncertaintyTheorem}. In the experiments, Eq.~\ref{eq:UncertaintyMain} was in fact an equality. As mentioned in Section~\ref{AnalyzingSec}, it is expected that in some cases the uncertainty $\DMT$ will be greater (satisfying the `$\geq$' in the equation) due to other factors that are not related to the measurement overhead.

\ifdefined\AddSpace\vspace{1mm}\fi
\subsection{Evaluating the Measurement Impact}
\ifdefined\AddSpace\vspace{1mm}\fi

In this experiment we evaluated the measurement impact $\DPM$ as a function of the computed measurement overhead $\Ov$. We used the IOAM setup, and varied the overhead per time unit by changing the IOAM sampling ratio. Two scenarios were tested: (a) the network is overprovisioned and the measurement impacts the data rate without loss impact, and (b) a 1~Mbps flow is forwarded over a 1~Mbps link without overprovisioning, thus causing packet loss.

\begin{figure}[!t]
  \centering
  \begin{subfigure}[t]{.24\textwidth}
  \centering
  \fbox{\includegraphics[height=4.8\grafflecm]{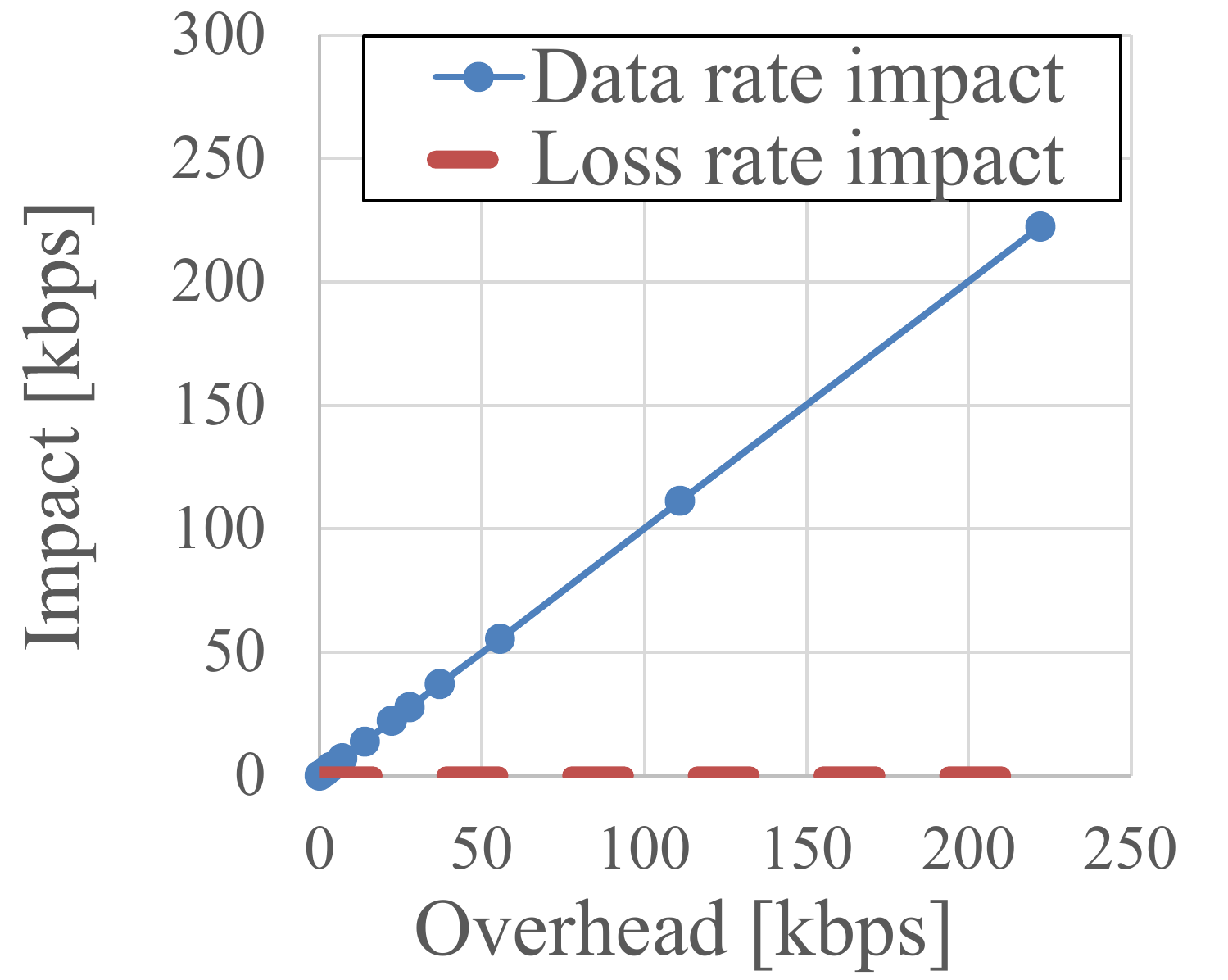}}
  \captionsetup{justification=centering}
  \ifdefined\CutSpace
	\caption{Overprovisioned link.}
	\else
	\caption{The measurement impact vs. the measurement overhead in an overprovisioned link.}
	\fi
  \label{fig:data}
  \end{subfigure}%
  \begin{subfigure}[t]{.24\textwidth}
  \centering
  \fbox{\includegraphics[height=4.8\grafflecm]{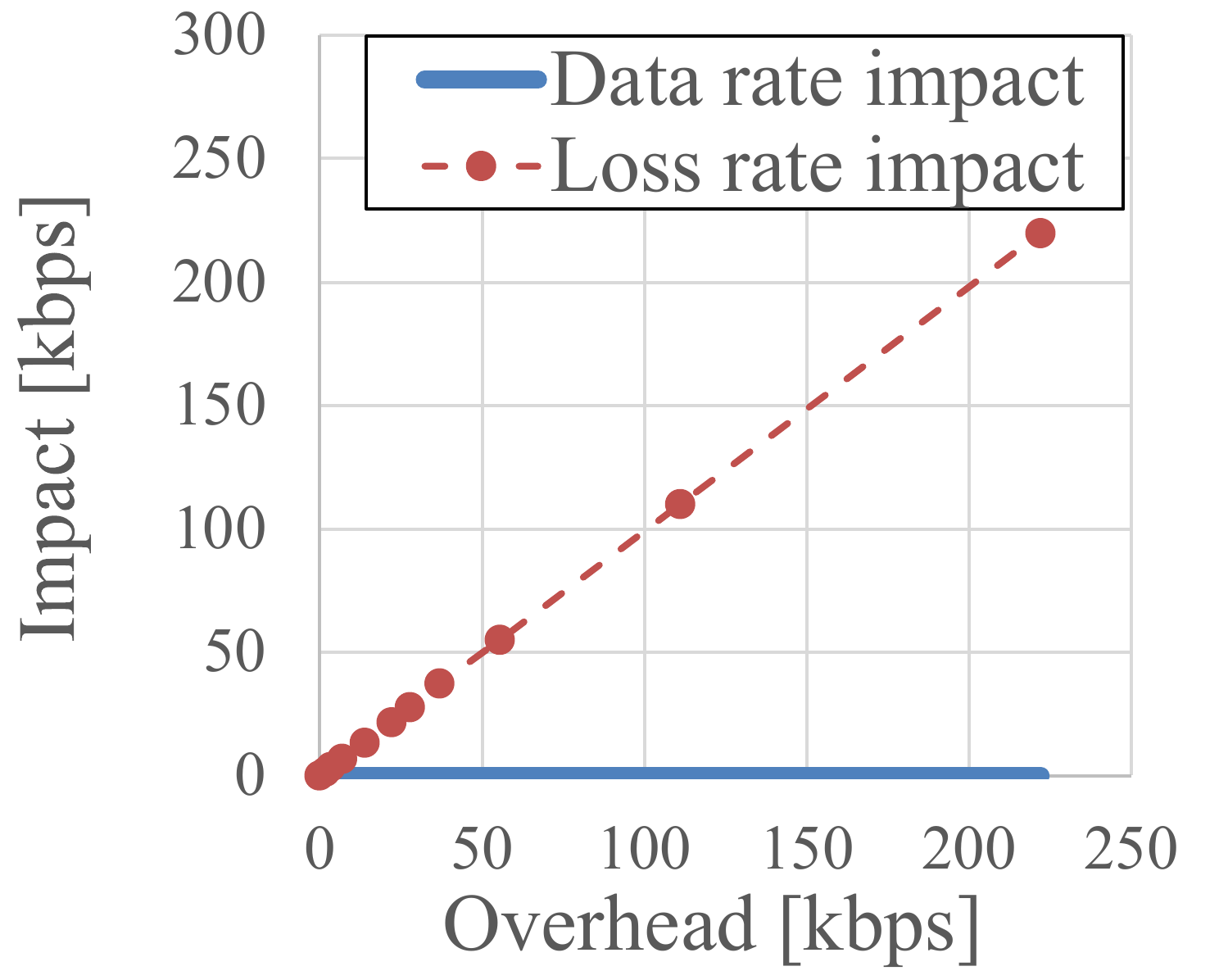}}
  \captionsetup{justification=centering}
  \ifdefined\CutSpace
	\caption{Without overprovisioning.}
	\else
  \caption{The measurement impact vs. the measurement overhead without overprovisioning.}
	\fi
  \label{fig:loss}
  \end{subfigure}%
  \ifdefined\CutSpace\vspace{-3mm}\fi
  \caption{The correlation between impact and overhead.}
	\vspace{-2mm}
  \label{fig:dataloss}
\end{figure}

The experimental results of Fig.~\ref{fig:dataloss} confirm the connection between the measurement impact and the overhead, stated in Lemma~\ref{RateMetricLemma}. Notably, our definition of \emph{impact} captures the cost of the measurement whether the network is overprovisioned or not; in the overprovisioned case the measurement impacts the data rate, whereas without overprovisioning the measurement impacts the loss rate.

\ifdefined\LongVersion

\begin{figure*}[!t]
  \centering
  \begin{subfigure}[t]{.33\textwidth}
  \centering
  \fbox{\includegraphics[height=6.5\grafflecm]{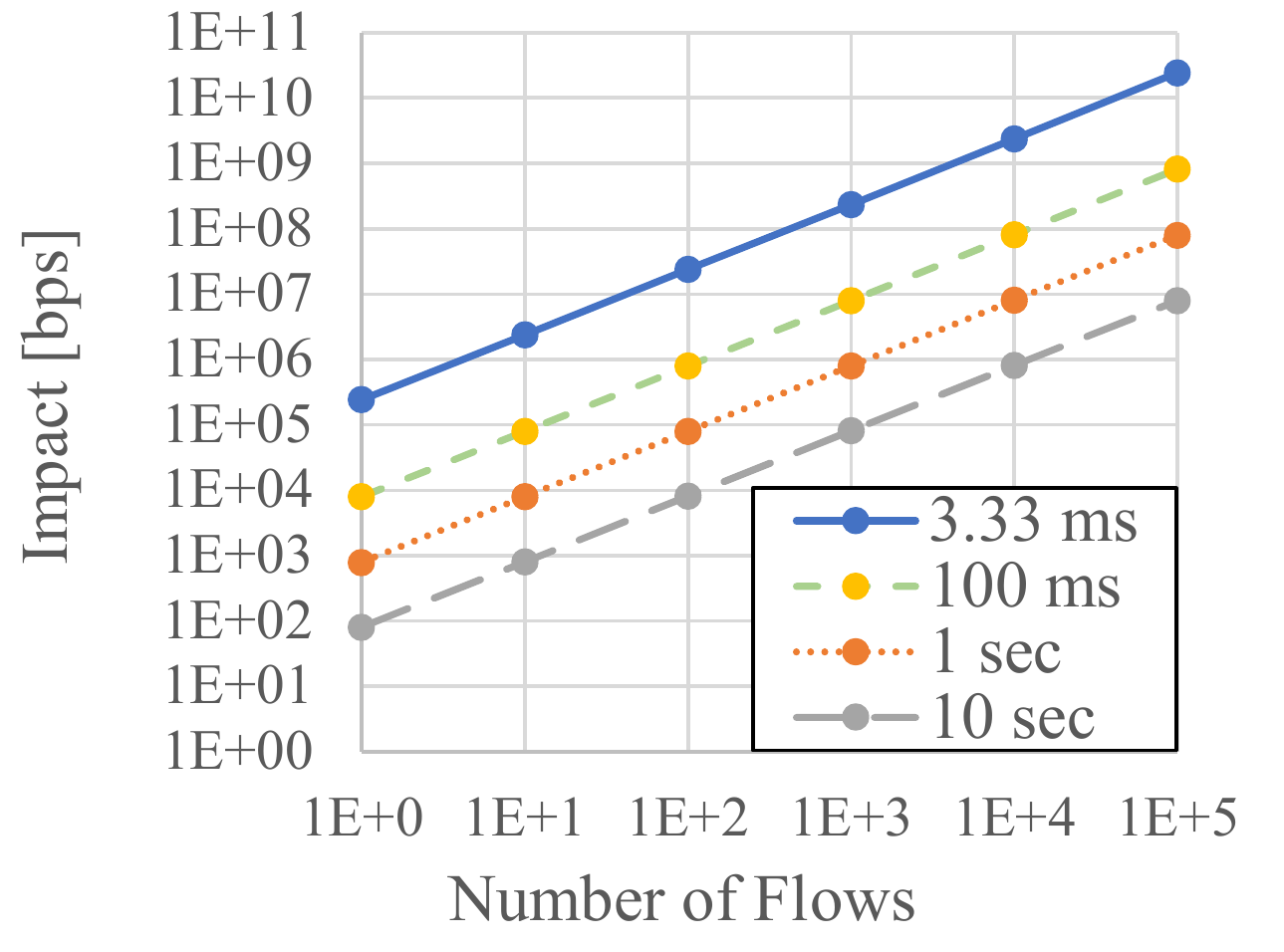}}
  \captionsetup{justification=centering}
  \caption{The impact as a function of the number \\ of flows (simulated). Each curve represents \\ a different measurement period.}
  \label{fig:ActiveImpactSim}
  \end{subfigure}%
  \begin{subfigure}[t]{.33\textwidth}
  \centering
  \fbox{\includegraphics[height=6.5\grafflecm]{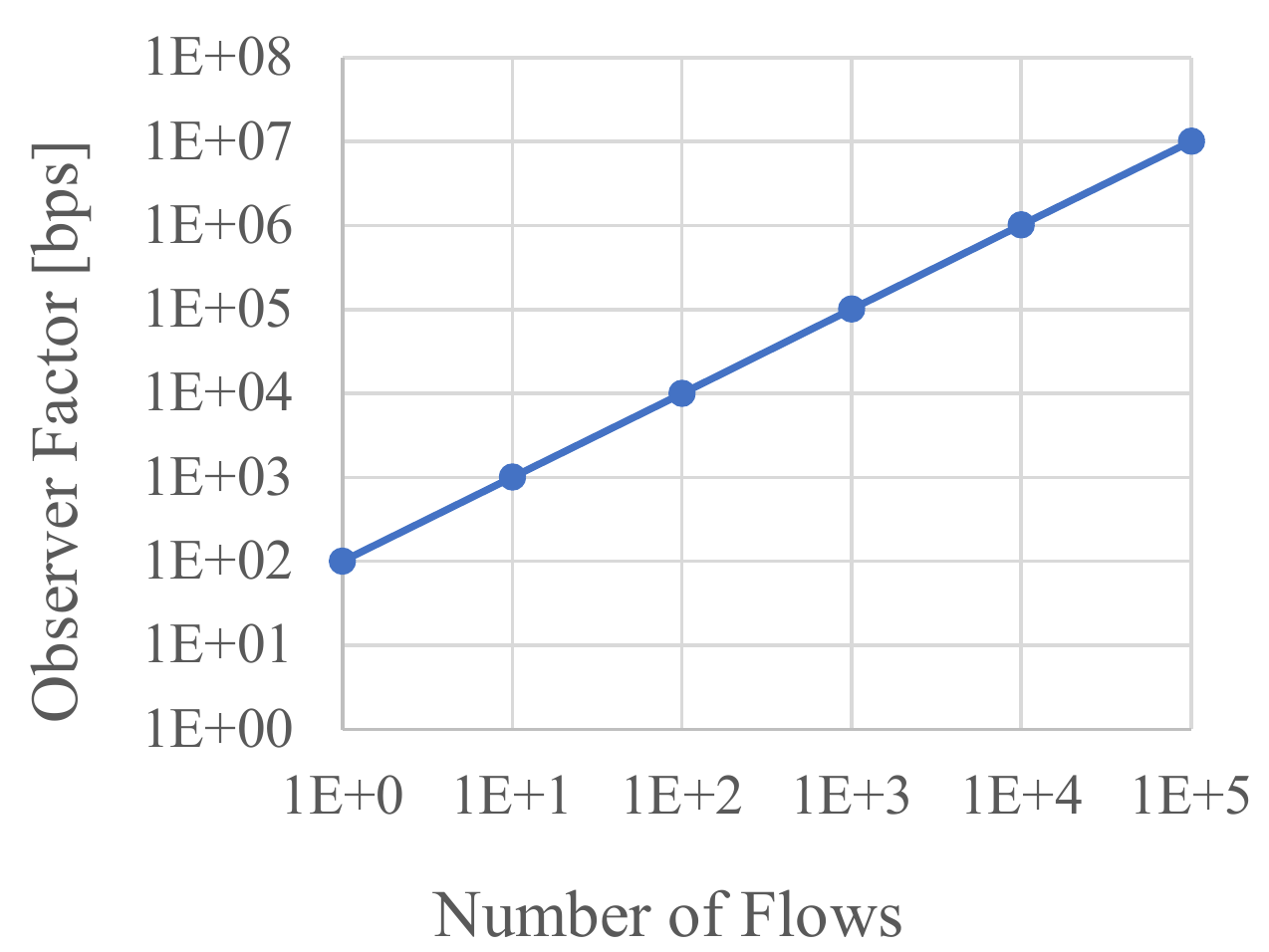}}
  \captionsetup{justification=centering}
  \caption{The observer factor as a function of the number of flows (computed).}
  \label{fig:FactorNoFlows}
  \end{subfigure}%
  \begin{subfigure}[t]{.33\textwidth}
  \centering
  \fbox{\includegraphics[height=6.5\grafflecm]{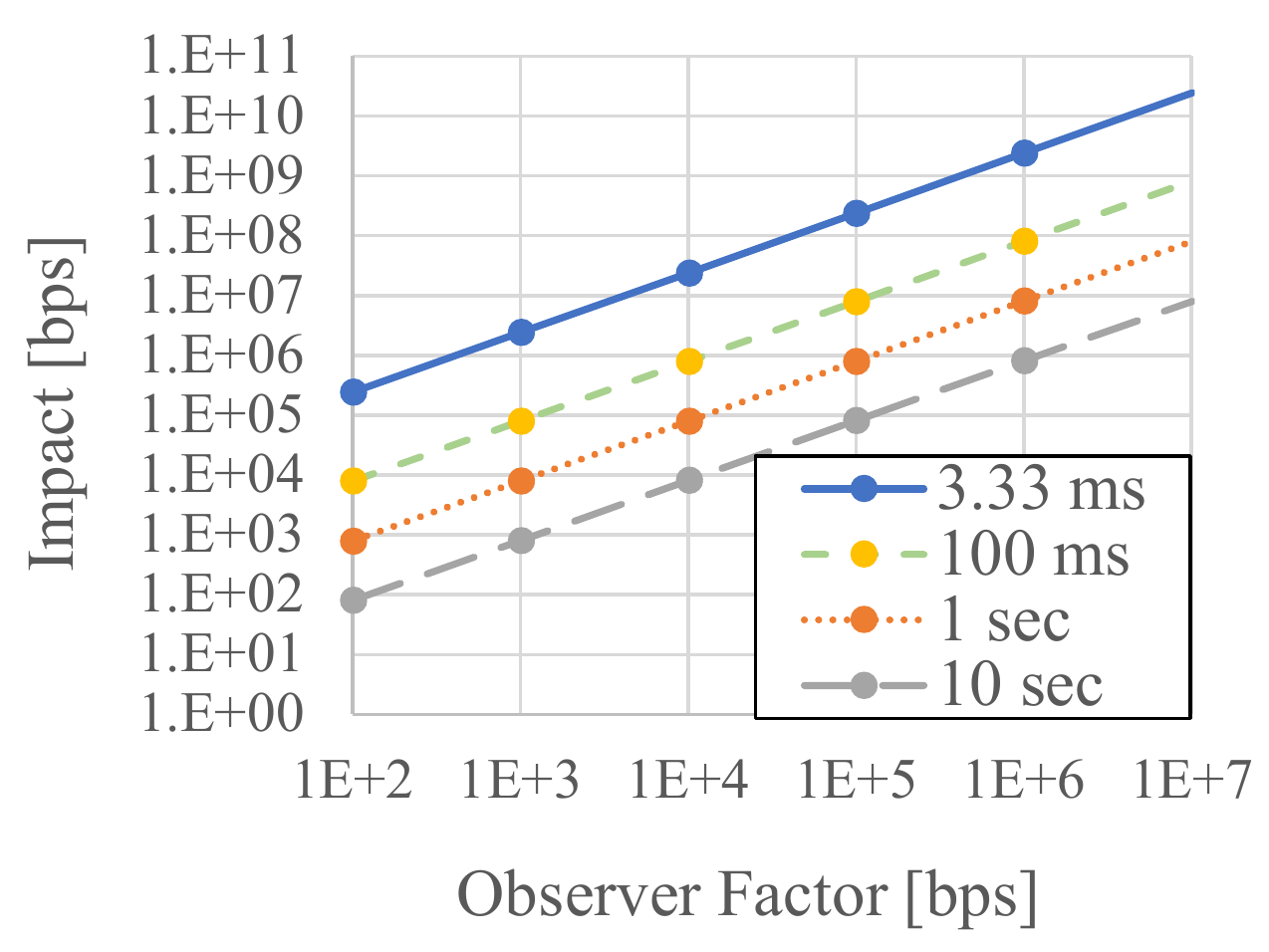}}
  \captionsetup{justification=centering}
  \caption{The impact as a function of the observer \\ factor (simulated). Each curve represents \\ a different measurement period.}
  \label{fig:ActiveImpactFactor}
  \end{subfigure}%
  \caption{The scaling of the observer effect as a function of the number of flows. Simulated for active measurement (CCM).}
  \ifdefined\CutSpace\vspace{-3mm}\fi
  \label{fig:scale}
\end{figure*}

\ifdefined\AddSpace\vspace{2mm}\fi
\subsection{Evaluating the Observer Effect Scaling}
\label{EvalScaling}
\ifdefined\AddSpace\vspace{1mm}\fi

In order to evaluate the observer effect at large scales we used a simulation environment. The simulation was implemented in Visual Basic, and its purpose was to evaluate the uncertainty relation and specifically the measurement impact in a network with a high data rate, and with a large number of flows. The advantage of a software-based simulation is that it allows analysis of large scales, as opposed to the emulation environments of the previous subsections, which ran in a virtualized environment on a conventional PC, and therefore allow up to tens of thousands of packets per second. 

The simulation environment mimicked the three scenarios of Section~\ref{ExperimentSetupSec}, but at a larger scale. The CCM scenario was simulated with a variable number of flows, between $1$ and $100,000$,\footnote{Supporting 10s of thousands of monitored flows on a single device is a reasonable use case in carrier networks~\cite{EricssonSPO1400,Huawei}, in which the CCM is used.} so that each flow was independently monitored by the CCM protocol. For each scale the simulation was repeated with different measurement period values: $3.33~ms$, $100~ms$, $1~sec$, and $10~sec$. The gNMI scenario was similarly simulated for a variable number of flows and with various values of measurement periods. The in-situ measurement was simulated at various data rates\footnote{The scaling factor of in-situ measurement is not the number of flows, since in-situ measurement uses per-packet overhead (subject to the sampling ratio). Instead, the scaling factor here was the data rate of the monitored data flow.} from $1~Mbps$ to $100~Gbps$, and the sampling ratio was varied from $1$ to $100$.

Fig.~\ref{fig:scale} presents simulation results that demonstrate the scaling behavior of Lemma~\ref{FlowLemma}. These results were produced from the active measurement (CCM) scenario. Using the notations of Lemma~\ref{FlowLemma}, Fig.~\ref{fig:ActiveImpactSim} depicts the impact, which is $\sum_{i=0}^{N-1} \Delta {P_T}_i (\tau)$, as a function of the number of flows $N$. Similar simulation runs were performed in the gNMI scenario and in the in-situ scenario, producing similar scaling behavior, as shown in Fig.~\ref{fig:scalePassiveInsitu}.

Fig.~\ref{fig:FactorNoFlows} illustrates the computed observer factor in the CCM scenario, $N \cdot \CPl$, as a function of the number of flows $N$. A comparison to Fig.~\ref{fig:ActiveImpactSim} demonstrates the correlation between the (simulated) impact and the (computed) observer factor. This correlation is then confirmed in Fig~\ref{fig:ActiveImpactFactor}, which shows the impact as a function of the observer 
\ifdefined\CutSpaceInfocom
factor.
\else
factor, demonstrating the uncertainty relation of Lemma~\ref{FlowLemma}.
\fi

These simulation results confirm the scaling behavior of the observer effect: the measurement impact scales with the size of the network. 
\ifdefined\CutSpaceInfocom
\else
Specifically, the simulations confirm the measurement impact scaling that lies in the uncertainty relation of Lemma~\ref{FlowLemma}.
\fi

\begin{figure}[htbp]
  \centering
  \begin{subfigure}[t]{.24\textwidth}
  \centering
  \fbox{\includegraphics[height=4.8\grafflecm]{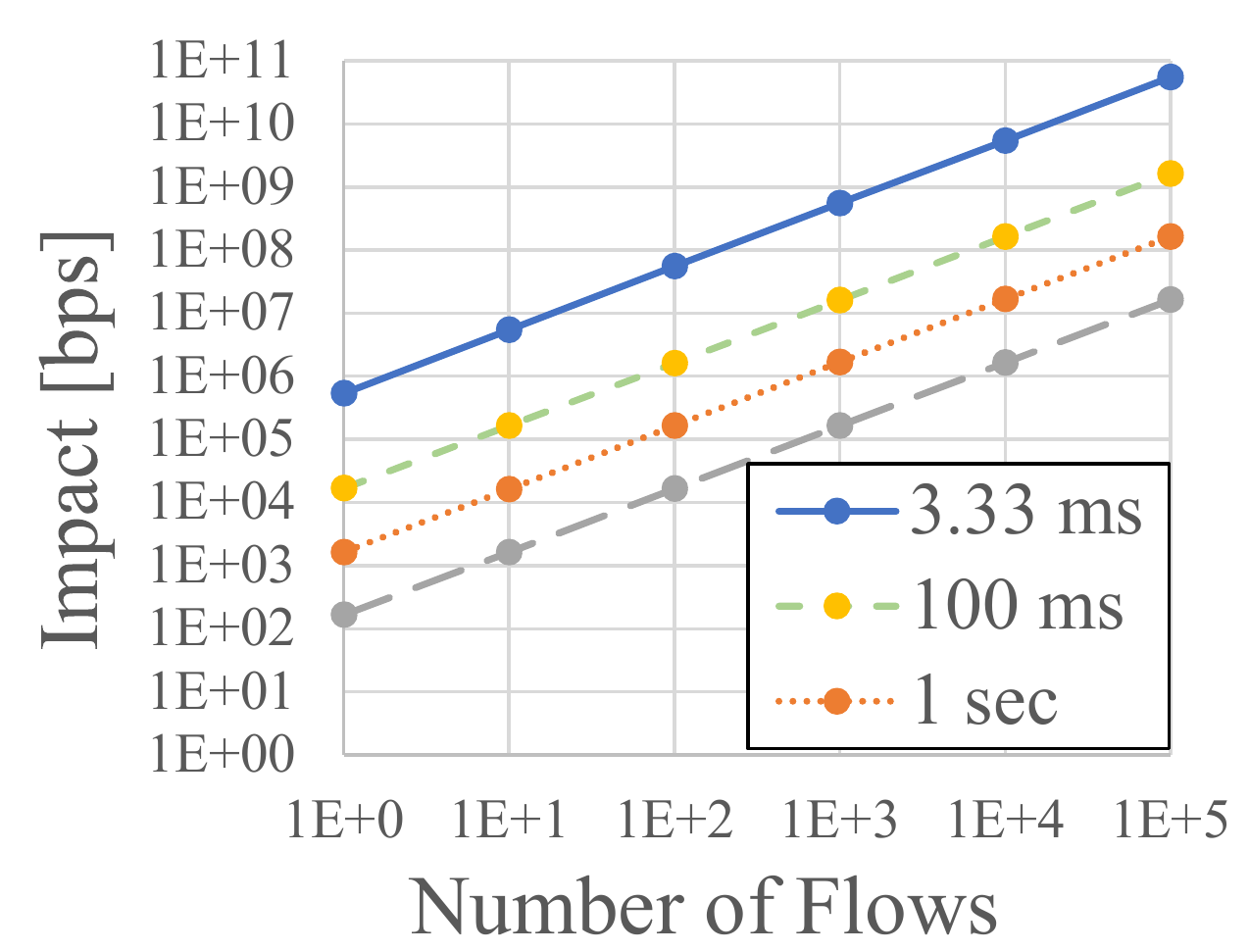}}
  \captionsetup{justification=centering}
\ifdefined\CutSpaceInfocom
  \caption{Passive measurement.}
\else
  \caption{Passive measurement impact as a function of the number of flows (simulated).}
\fi
  \label{fig:PassiveSim}
  \end{subfigure}%
  \begin{subfigure}[t]{.24\textwidth}
  \centering
  \fbox{\includegraphics[height=4.8\grafflecm]{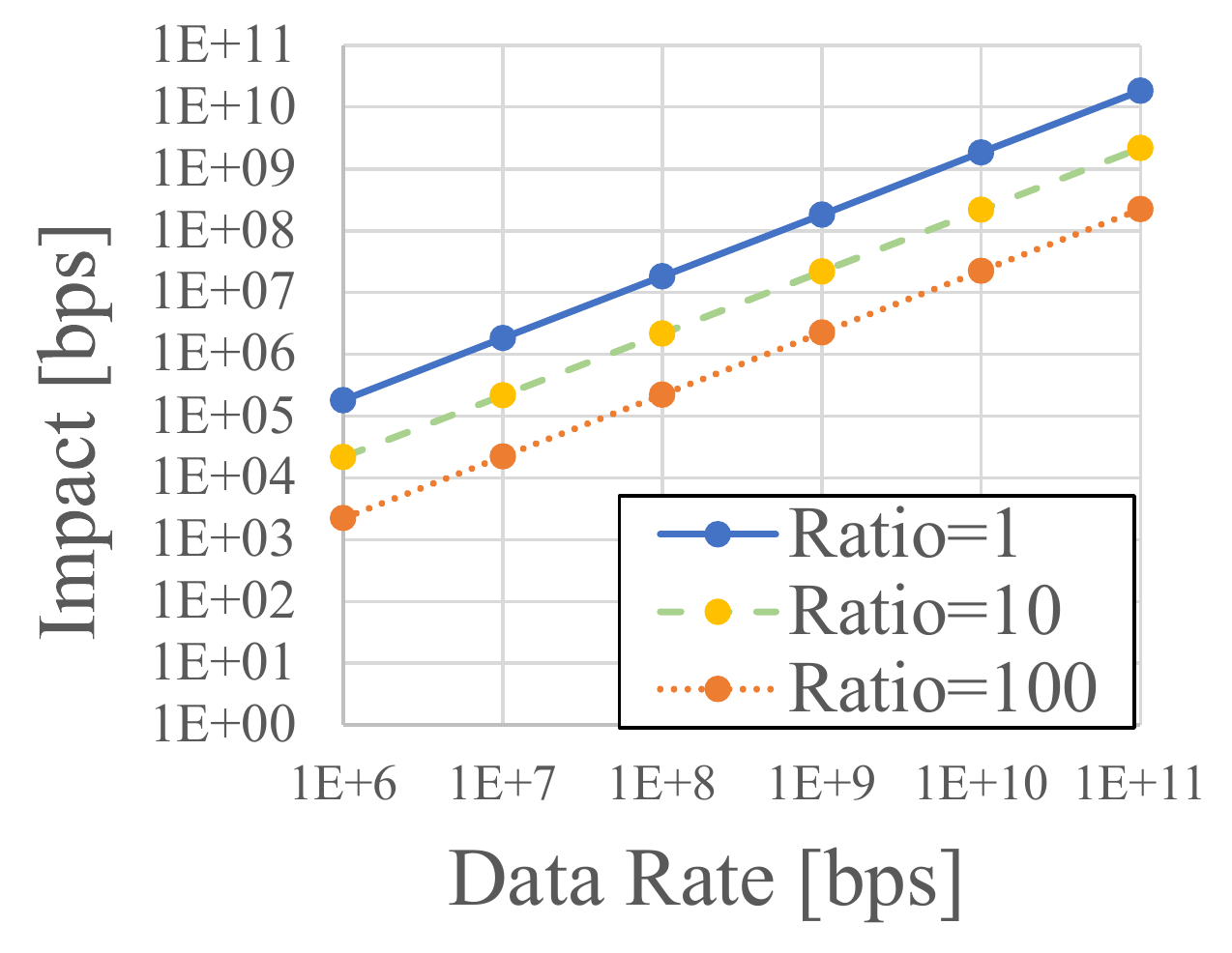}}
  \captionsetup{justification=centering}
\ifdefined\CutSpaceInfocom
  \caption{In-situ measurement.}
\else
  \caption{In-situ measurement impact as a function of the data rate (simulated).}
\fi
  \label{fig:InsituScaleSim}
  \end{subfigure}%
  \ifdefined\CutSpace\vspace{-3mm}\fi
  \caption{Scaling of the measurement impact (simulated) in passive and in in-situ measurement. Each curve represents a different measurement period (a) or sampling ratio (b).}
  \ifdefined\CutSpace\vspace{-6mm}\fi
  \label{fig:scalePassiveInsitu}
\end{figure}

\fi

\ifdefined\CutBeyond
\else
\ifdefined\AddSpace\vspace{1mm}\fi
\section{Beyond the Observer Effect}
\ifdefined\AddSpace\vspace{1mm}\fi
\label{BeyondSec}
In this paper we analyzed the impact of measurement overheads on the network performance. However, 
measuring and observing the network does not only add overhead traffic. It may also, in some cases, consume processing power in networking devices including hosts, switches and routers. This processing overhead may affect the data plane processing, the control plane processing, or both. In either case, this may impact the network performance. 
Moreover, the measurement overhead affects not only the network resources, but also the computing and storage resources of the hosts that take part in the measurement protocols and the server(s) that monitor and analyze the performance. 

Another aspect that should be considered is fate-sharing: the traffic that is used for measuring the network must follow the same path and forwarding policy as the measured user traffic. However, \inband\ measurement may increase the size of monitored packets, possibly causing the forwarding behavior to be different from that of the original user traffic. Thus, the measurement protocol may cause data traffic to be forwarded differently than without the measurement. Furthermore, fate-sharing is not guaranteed if \inband\ measurement is applied only to a subset of the traffic.

Measurement inaccuracy may also be affected by other factors. One factor may be security aspects; attackers may maliciously tamper with network measurements in order to create a false illusion of a network problem or to hide the existence of one~\cite{rfc7276}. Another factor that may affect the measurement accuracy in wide area networks is net neutrality; some service providers have been known to detect and assign higher priority to speed test traffic~\cite{dischinger2010glasnost}. Generally speaking, measurements may be affected by intermediate nodes that do not `play nice'. These considerations and other aspects of net neutrality have not been in the focus of the current paper, and are worthy of consideration.




\fi

\section{Related Work}
\ifdefined\AddSpace\vspace{1mm}\fi
\begin{sloppypar}
The uncertainty principle and the observer effect have been widely discussed in the literature (e.g.,~\cite{heisenberg1930physical,heisenberg1958physics}). In computer science, the observer effect was analyzed in the context of computing performance~\cite{mytkowicz2008observer}, but to the best of our knowledge the current paper is the first to consider and analyze the observer effect in the context of communication networks. The literature is rich with work about network measurement methods: passive (e.g.,~\cite{Cisco,rfc7011,yu2013flowsense,zhu2015packet}), active (e.g.,~\cite{IEEE802.1ag,Y1731,mittal2015timely,zhu2015packet}), and \inband\ measurement~\cite{kimband,INT,IOAM,kumar-ippm-ifa-01}.

The fact that network measurement incurs overhead is common knowledge (see ~\cite{wang2014timing,eriksson2008network,popescu2017ptpmesh,zhu2015packet,rfc7799}), as is the fact that there is a tradeoff between a measurement's accuracy and its impact~\cite{soule2005traffic,li2019hpcc,dischinger2010glasnost}. We are, however, unaware of prior formal and \textit{quantitative} analyses of the reciprocal relation between the measurement overhead and the measurement accuracy (and, in particular, in analogy to principles from quantum mechanics).
\end{sloppypar}

\ifdefined\AddSpace\vspace{1mm}\fi
\begin{sloppypar}
\section{Conclusion}
\ifdefined\AddSpace\vspace{1mm}\fi
We presented and formalized an observer effect for computer networks, which captures the interplay between network performance and its measurement. The observer effect yields a delicate tradeoff between the required measurement granularity and the imposed performance impact, affecting both the networking and computing resources.

We believe that the cost efficiency of any network measurement and monitoring method should be evaluated with this granularity-impact tradeoff in mind. 
We view our network observer effect/factor framework as a first step towards better understanding and theoretical evaluation of efficient network measurement. As networks continue to grow in size, requiring increasingly high visibility and transparency, the observer effect provides a way to formally capture and analyze how measurement impacts performance.
\end{sloppypar}



\bibliographystyle{ieeetr}
\bibliography{Uncert}
\end{document}